\documentclass[a4paper]{llncs}

\usepackage{paper}

\title{Optimal Metric Search Is Equivalent to the
    Minimum Dominating Set Problem}

\author{Magnus Lie Hetland}
\institute{Norwegian University of Science and Technology,
\email{mlh@ntnu.no}%
}

\begin{document}
\maketitle

\begin{abstract}
    In metric search, worst-case analysis is of little value, as the search
    invariably degenerates to a linear scan for ill-behaved data.
    Consequently, much effort has been expended on more nuanced descriptions
    of what performance might in fact be attainable, including heuristic
    baselines like the AESA family, as well as statistical proxies such as
    intrinsic dimensionality. This paper gets to the heart of the matter with
    an exact characterization of the best performance actually achievable for
    any given data set and query. Specifically, linear-time
    objective-preserving reductions are established in both directions between
    optimal metric search and the minimum dominating set problem, whose greedy
    approximation becomes the equivalent of an oracle-based AESA, repeatedly
    selecting the pivot that eliminates the most of the remaining points. As
    an illustration, the AESA heuristic is adapted to downplay the role of
    previously eliminated points, yielding some modest performance
    improvements over the original, as well as its younger relative iAESA2.
\end{abstract}

\section{Introduction}

Mapping out the complexity of a computational problem is generally a
two-pronged affair. On the one hand, there will be algorithms solving the
problem, whose performance is evaluated theoretically or empirically,
providing ever-tightening pessimistic bounds on what is possible. On the other
hand, there may be \emph{lower} bounds, based on reasonable
complexity-theoretical assumptions, as in the case of edit distance, for
example~\cite{Backurs:2015}, or on reasoning about the fundamentals of the
computational model, as in the case of sorting~\cite{Ford:1959}. The endgame
is when these bounds meet, showing some algorithm to be optimal.

Such bounds generally apply to the worst case, as the best-case performance
tends to be trivial. For metric search, however, both the best case and the
worst are quite uninformative. For a range query, one could always construct
an input where examining a single object is enough---\kern.6ptor one where
there is no escaping a full linear scan. The main thrust of research
attempting to describe what performance is possible has thus been directed
toward empirical baselines like the AESA
family~\cite{Vidal:1986,Figueroa:2010} and statistical hardness measures such
as intrinsic dimensionality~\cite{Chavez:2001},\footnote{Other measures
include the distance exponent~\cite{Traina:2000} and the ball-overlap
factor~\cite{Skopal:2007}.} or in some cases restricting the type of structure
studied, to permit a more nuanced analysis~\cite{Pestov:2012}.

It is, however, possible to describe \emph{exactly} what performance is
attainable for a given data set and query, as I show in what follows. The main
equivalence result, between metric search and dominating sets, provides just
such a description, i.e., the lowest number of distance computations that can
resolve the query. This performance will not, in general, be attainable
without some lucky guesses, but it \emph{is} attainable.
In addition, it is possible to give a bound on how close to this performance
a polytime algorithm may come in the worst case, under reasonable complexity
assumptions. The bound is tight for a sufficiently precise pivot selection
heuristic, i.e., one that is able to predict which point will eliminate the
most of the remainder, if used as a pivot.

In the AESA method, the index is a distance matrix, and search alternates
between heuristically selecting points close to the query and eliminating
remaining objects that are shown to be irrelevant.
The results in this paper are based on an idea developed by Ole
Edsberg,\footnote{Personal communication, July 2012} which involves computing
an \emph{elimination} matrix for a given query, with which one may implement
an ``oracle AESA,'' selecting pivots greedily based on elimination power,
rather than on similarity to the query object. I build on this idea,
establishing equivalence to the minimum dominating set problem.\footnote{Note
that the reductions are to and from two different versions of the dominating
set problem (the directed and undirected version, respectively). At the price
of slightly looser bounds, one could stick with just one of these.}
The main results and contributions of the paper are summarized in the
following.

\textit{Reduction to domination.} \Cref{sec:pivopt,sec:elimination}
establish a linear-time objective-preserving reduction from the problem of
resolving metric range queries (and certain $k$NN queries) with as few
distance computations as possible to that of finding minimum dominating sets
in directed graphs. This reduction applies to an \emph{offline} variant of
metric search, where all query--\kern.5pt object distances are already known.
It does, however, make it possible to compute the exact optimum
attainable
for the online version as well.
Some experimental results are provided as an illustration.

\textit{Reduction from domination.} \Cref{sec:hardness} describes a
reduction in the other direction,
from the dominating set problem in \emph{undirected} graphs to minimizing
distance computations,
establishing the hardness of metric search.
While it may in many cases still be feasible to determine the optimum using
efficient solvers of various kinds, this does mean that under reasonable
complexity-theoretical assumptions, no search method can, in general,
\emph{guarantee} attaining this optimum.

The reduction preserves the objective value, and for range search, the number
of data objects equals the number of vertices, which means that
inapproximability results for the dominating set problem carry over to metric
search, with approximation bounds for the former applying to the performance
of the latter, i.e., the number of distance computations. Thus, for range
search, one cannot even expect to get closer than within a $\log$-factor of
the optimum.

\textit{AESA and greedy approximation.} Because the objective is preserved
also in reducing \emph{to} domination, and the number of objects equals
the number of vertices, \emph{approximability} results also translate, meaning
that in principle the standard greedy selection strategy would yield the best
feasible metric range search algorithm (or very close to it), in terms of
distance computations in the worst case.\footnote{This is the worst case
\emph{given} that the optimal number of distance computations is some
value $\gamma$, not the more general, non-informative worst-case of
$\Omega(n)$.} As discussed in \cref{sec:aesa}, the greedy approach corresponds
to the AESA family of algorithms, given the right selection heuristic, i.e.,
one that accurately estimates the elimination power of a potential pivot,
among the remaining objects. An exact estimate here is, of course, not
possible without knowing the query--\kern.5pt pivot distances, but this
correspondence does demonstrate that, in the limit, AESA is, indeed, as good
as it gets. As an illustration, inspired by the greedy approximation,
\emph{greedy} AESA (gAESA) is proposed, taking into account which points
remain to be eliminated.

\section{Pivoting Is, of Course, Optimal}
\label{sec:pivopt}

A \emph{range search} using a metric $\delta$ over a set $X$ means finding all
points $x\in X$ within some search radius $r$ of a given query point $q$,
i.e., all points $x$ for which $\delta(q,x)\leq r$.
Given the distances between a query $q$ and a set $P$ of \emph{pivots}, the
distance $\delta(q,x)$ for any point $x$ is bounded as follows:
\begin{equation}
    \textstyle
    \max_{p\in P} | \delta(q,p) - \delta(p,x) |
    \leq
    \delta(q,x)
    \leq
    \min_{p\in P} \delta(q,p) + \delta(p,x)
    \label{eq:bounds}
\end{equation}
Leaving $q$ and $P$ implicit, we may refer to the lower and upper bounds as
$\ell(x)$ and $u(x)$, respectively.
If our search radius falls outside this range, there is no need to compute
$\delta(q,x)$; either the radius is small enough that we simply eliminate $x$
($r < \ell(x)$), or it is great enough that $x$ is ``eliminated'' by adding it
to the search result, sight unseen ($r \geq u(x)$).

This very direct approach
of using exact, stored distances $\delta(p,x)$, \emph{pivoting}, is the gold
standard for minimizing the number of distance computations needed. Other
approaches, which all involve coarsening the stored information in some way,
may reduce the computational resources needed to eliminate candidate objects,
but it should be obvious that they cannot require fewer distance computations.
As the following lemma shows, the lower and upper bounds are necessarily valid
values for $\delta(q,x)$, so if $\ell(x)\leq r \leq u(x)$, $x$ cannot safely
be eliminated.

\begin{lemma}
    \label{lem:stillmetric}
    Let $(X,\delta)$ be a metric space, with $X=\{p_1,\dots,p_m,q,z\}$, and
    let the distances $\delta_1,\delta_2:X\times X\to\mathds{R}_{\geq 0}$ be
    defined as follows:
    \[
    \delta_1(x,y)=
    \begin{cases}
        \textstyle
             \max_i |\delta(q,p_i)-\delta(p_i,z)| & \text{if
             $\{x,y\}=\{q,z\}$\rm;}
             \\
             \delta(x,y) & \text{otherwise\rm.}
    \end{cases}
    \]
    \[
    \delta_2(x,y)=
    \begin{cases}
        \textstyle
            \mathrlap{\min_i}\phantom{\max_i |}
            \delta(q,p_i)+\delta(p_i,z)\phantom{|}
             & \text{if
             $\{x,y\}=\{q,z\}$\rm;}
             \\
             \delta(x,y) & \text{otherwise\rm.}
    \end{cases}
    \]
    Then $\delta_1$ is a pseudometric and $\delta_2$ is a metric. If
    $\delta(q,p_i)\neq\delta(p_i,z)$ for some $i$, or if $q=z$, then
    $\delta_1$ is a metric.
\end{lemma}
\begin{proof}
    We have $\delta_j(x,y)=\delta_j(y,x)$ and $\delta_j(x,x)=0$, for
    $x,y\in X, j\in\{1,2\}$. We also have $\delta_2(x,y)=0\implies x=y$, and
    if $\delta(q,p_i)\neq\delta(p_i,z)$ for some $i$, or if $q=0$, then
    $\delta_1(x,y)=0\implies x=y$.
    We have $\delta_1(q,z)\leq\delta(q,z)$, so triangularity
    can only be broken for $\delta_1$
    in the cases
    $\delta_1(q,p_k)\leq\delta_1(q,z)+\delta_1(z,p_k)$
    or
    $\delta_1(z,p_k)\leq\delta_1(z,q)+\delta_1(q,p_k)$, for some $k$. Consider
    the first of these.
    We maximize over $i$, so we need only show the following for \emph{some}
    choice of $i$:
    \begin{equation}
        \delta(q,p_k) \leq |\delta(q,p_i)-\delta(p_i,z)| + \delta(z,p_k)
    \end{equation}
    This is satisfied for $k=i$.
    The other case is handled symmetrically.
    For $\delta_2$, we have $\delta(q,z)\leq\delta_2(q,z)$, so
    triangularity can only be broken in $\delta_2(q,z)\leq
    \delta_2(q,p_k)+\delta_2(p_k,z)$. We minimize over $i$, so this need only
    hold for \emph{some} choice of $i$, and again we may choose $i=k$,
    producing an equation.
\end{proof}
\begin{corollary}
    \label{cor:range}
    No search method can resolve a metric range query with fewer distance
    computations than pivoting.
\end{corollary}
\begin{proof}
    From \cref{lem:stillmetric}, we know that after a set of distance
    computations making pivots $p_1,\dots,p_m$ available, the pivoting bounds
    are \emph{tight}; if pivoting cannot eliminate an object, no method can
    safely do so.
    (Note that an adversary would be free to let $q=z$ in the case where
    $\delta(q,z)=\ell(z)=0$, ensuring that we are indeed dealing with a metric
    space.) And given that no method can eliminate more objects than pivoting
    for any distance count, no method can eliminate all the objects with a
    lower distance count than pivoting.
\end{proof}

In other words, any method using fewer distance computations than pivoting
could be made to fail by an adversary in charge of the data set. This argument
covers range queries, and it is not hard to translate it to the $k$NN case,
where the $k$ nearest neighbors of $q$ are sought, as long as the result set
is uniquely determined. A radius must then exist, separating the $k$ nearest
neighbors from the others, and the tightest possible upper bound on this
radius is the maximum of the $k$ lowest pivoting bounds we have. Pivoting must
then be able to eliminate all points outside this radius, or our adversary
might strike again. The following corollary covers the more general case.

\begin{corollary}
    \label{cor:knn}
    No search method can resolve a metric $k$NN query with fewer distance
    computations than pivoting.
\end{corollary}
\begin{proof}
We need to establish $\delta(q,x)\leq\delta(q,y)$ for every $x$
in the result and every $y$ outside it. Assume that, given some pivot set $P$,
there is one such inequality that cannot be established by pivoting, i.e.,
$u(x)>\ell(y)$. An adversary could then ensure $\delta(q,x)>\delta(q,y)$, as
follows. First, let $\delta(q,x)=u(x)$.
The only effect on the valid range for $\delta(q,y)$ is found in the lower
bound $\delta(q,y)\geq|u(x)-\delta(x,y)|$. If $\delta(x,y)\leq
u(x)$, then the relevant lower bound is $u(x)-\delta(x,y)$, which is
\emph{strictly} less than $\delta(q,x)=u(x)$ (because $\delta(x,y)>0$, as
$x\neq y$), and so it is still possible to have $\delta(q,y)<\delta(q,x)$.

If, however, $\delta(x,y)>u(x)$, the relevant lower bound is
$\delta(x,y)-u(x)$. Let $p$ be the pivot that produced the pivoting bound
$u(x)$. We then have:
\begin{align*}
\delta(x,y) - u(x) &= \delta(x,y) - \bigl(\delta(q,p) + \delta(p,x)\bigr) \\
                   &= \bigl(\delta(x,y)-\delta(p,x)\bigr) - \delta(q,p)
                   \leq \delta(p,y)-\delta(q,p)
                   \leq \ell(y)
\end{align*}
In other words, $\delta(q,y)=\ell(y)$ is still a valid choice for the
adversary, yielding the desired $\delta(q,y)<\delta(q,x)$.
\end{proof}

\noindent
The upshot
is that the optimal distance count (for range and $k$NN queries)
can be found by considering only elimination \emph{using individual pivots}.

The range and $k$NN search modes are closely related, and yet there are cases
where they behave quite differently, as shown in \cref{fig:knnvsrange}.
\begin{figure}
\def\r{2pt}
\newcommand{\pt}[1]{circle (\r) node[below=\r] {#1}}
\hfill
\subfigure[Range wins ($k=2$)]{\label{subfig:rangewins}
\begin{tikzpicture}

    \draw[step=0.5,gray,very thin] (0,0) grid (3,3);

    \draw[fill] (0,0) \pt{$q$}
    (1,1) \pt{$p$}
    ;

    \foreach \i in { 1, ..., 5} {
        \draw[fill] (0.5+0.5*\i,3.5-0.5*\i) \pt{$x_\i$};
    }

\end{tikzpicture}
}
\hfill
\subfigure[$k$NN wins ($k=1$)]{\label{subfig:knnwins}
\begin{tikzpicture}

    \draw[step=0.5,gray,very thin] (0,0) grid (3,3);

    \draw[fill] (0,0) \pt{$q$}
    (3,3) \pt{$p$}
    ;

    \foreach \i in {1, ..., 5} {
        \draw[fill] (0.5+0.5*\i,3.5-0.5*\i) \pt{$x_\i$};
    }

\end{tikzpicture}%
}\hfill\mbox{}%
\caption{Differences between range search and $k$NN in the presence of ties
    for the $k$th position, using $(\mathds{R}^2,L_1)$. In both
    configurations, we have $r=8$. In (a), a range search need only compute
    $\delta(q,p)$, while $k$NN much also compute $\delta(q,x_i)$ for all but
    one of the $x_i$. In (b), the $k$NN search need only compute $\delta(q,p)$
    and $\delta(q,x_i)$ for one of the $x_i$, while a range search must
    compute all distances $\delta(q,\blank)$%
}\label{fig:knnvsrange}
\end{figure}
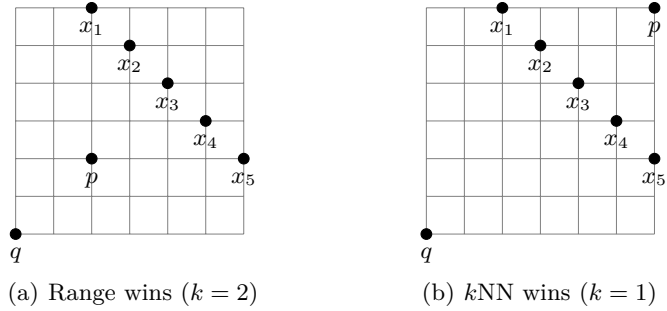

\noindent
It is, however, possible to establish some correspondence between the two,
when the $k$NN result set is uniquely determined.

\begin{lemma}
    If the $k$NN result is uniquely determined, the optimum number of distance
    computations for $k$NN is no worse than for a range search with the
    smallest possible $k$NN radius, even if the radius is unknown initially.
    Furthermore, there is a radius for which range queries and $k$NN will
    produce the same search result using the same number of distance
    computations.
\end{lemma}
\begin{proof}
When the $k$NN result is unique, there is a radius $r$ corresponding to the
$k$ resulting objects. Resolving a range query with radius $r$ must
necessarily yield upper bounds of at most $r$ for returned objects and lower
bounds greater than $r$ for the remainder. These same bounds can also be used
to separate the $k$ nearest objects from the remainder, \emph{without a
specified radius}, so $k$NN cannot require more distance computations.

Conversely, consider a $k$NN query. By \cref{cor:knn}, no method
requires fewer distance computations than pivoting, so in the optimal case we
will have actual distance bounds available, strictly separating the $k$
nearest from the remainder. Any range query with a search radius falling
between the upper and lower bounds can then be resolved with the same number
of distance computations.
\end{proof}

\noindent
It \emph{is} possible to increase the radius such that a range query would
require additional distance computations, while still just returning $k$
objects (cf.~\cref{fig:largerradius}).

\section{Elimination as Domination}
\label{sec:elimination}

Given a (directed) graph $G=(V,E)$, a vertex $u$ is said to \emph{dominate}
another vertex $v$ if the graph has an edge from $u$ to $v$. The
(\emph{directed}) \emph{minimum dominating set problem} involves finding a set
$D\subseteq V$ of minimum cardinality, such that every vertex $v\in V\setminus
D$ is dominated by some vertex $u\in D$. We call $\gamma(G)=|D|$ the
(\emph{directed}) \emph{domination number} of $G$.

For a given range query, computing the distance to a point may eliminate one
or more other points. There are no interactions between such eliminations
(see \cref{sec:pivopt}),
so an exhaustive listing of the potential eliminations gives us all the
relevant information needed to determine which points to examine and which to
eliminate. This corresponds to a directed graph---the \emph{elimination
graph}---whose minimum dominating set is the smallest pivot set, and thus the
minimum number of distance computations, needed to resolve the query
(cf.~\cref{fig:elimgraph}).
\begin{proposition}
    There is a linear-time reduction from the metric range search problem to
    the directed minimum dominating set problem, which preserves the objective
    values of the solutions exactly.
    \qed
\end{proposition}

If the result of a $k$NN query is uniquely determined, and we ignore
elimination based on upper bounds (usually done in practice), the number of
distance computations correspond to a range query with the smallest $k$NN
radius.\footnote{Optimal $k$NN \emph{with} upper bounds does not map as
cleanly to dominating sets.}
\begin{figure}
\begin{floatrow}
\ffigbox{%
\def\r{2pt}
\begin{tikzpicture}
\begin{scope}
\clip (-.5\linewidth + 1.25cm + .5\pgflinewidth,-1.75)
    rectangle (.5\linewidth + 1.25cm - .5\pgflinewidth,1.75);
\fill[black!5] (-10cm,-10cm) rectangle (10cm,10cm);

\draw[black!20, fill=white] (0,0) circle[radius=2.5];
\draw[black!20] (0,0) circle[radius=1.5];

\draw (-35:1.5) node[fill=white] {$r_1$};

\begin{scope}
    \draw (-25:2.5) node[fill=black!5] {$r_2$};
\end{scope}

\begin{scope}
    \clip (0,0) circle[radius=2.5];
    \draw (-25:2.5) node[fill=white] {$r_2$};
\end{scope}

\draw
    (0,0) node[querynode] (q) {}
    (1,0) node[plainnode] (o1) {}
    +(2,.5) node[plainnode] (p) {}
    +(2,-.5) node[plainnode] (o2) {}
    ;

\draw
    (q.south) node[below, font=\footnotesize] {$q$}
    (o1.south) node[below, font=\footnotesize] {$x_1$}
    (o2.south) node[below, font=\footnotesize] {$x_2$}
    (p.south) node[below, font=\footnotesize] {$p$}
    ;

\draw[densely dashdotted] (p) circle[radius=1];

\end{scope}
\draw (current bounding box.south west) rectangle (current bounding box.north east);

\end{tikzpicture}%
}{%
\caption{The nearest neighbor can be determined by examining $x_1$ and $p$, as
we then have $u(x_1)<\ell(x_2)$. Range search with $r_1$ can be resolved
similarly, but using $r_2$ requires three distance computations, while still
returning the single nearest neighbor}%
\label{fig:largerradius}%
}%
\ffigbox{%
\begin{tikzpicture}
\begin{scope}
\clip (-.5\linewidth + .5\pgflinewidth,-1.75)
    rectangle +(\linewidth - \pgflinewidth, 3.5);

\fill[black!5] (-10cm,-10cm) rectangle (10cm,10cm);

\node{%
\begin{tikzpicture}[scale=.47]

    \draw (3.732886157216435 cm, 0.3549838538817218 cm) node[plainnode] (1) {};

    \draw (7.939504727723011 cm, 0.5322674524966186 cm) node[plainnode] (2) {};

    \draw (1.1588123954315703 cm, 1.622988764034955 cm) node[plainnode] (3) {};

    \draw (9.910698259877051 cm, 2.4882154092025037 cm) node[plainnode] (4) {};

    \draw (2.5018709565510746 cm, 4.174377590352696 cm) node[plainnode] (5) {};

    \draw (6.539896044904356 cm, 3.7207041770509033 cm) node[plainnode] (6) {};

    \draw (2.241320885523173 cm, 6.288158229306186 cm) node[plainnode] (7) {};

    \draw (5.866025777076387 cm, 6.393193053894219 cm) node[plainnode] (8) {};

    \draw (8.137267087234004 cm, 6.7759220129894855 cm) node[plainnode] (10) {};

    \draw (9.781976957678093 cm, 4.8523612977115445 cm) node[querynode] (q) {};
    \draw (q.south) node[below, font=\footnotesize] {$q$};

    \begin{pgfonlayer}{background}

    \draw[fill=white, draw=black!20, overlay]%
        (q.center) circle[radius=5.127083089556588];

    \end{pgfonlayer}

    \draw (1) edge[shorten <=2pt, shorten >=2pt, stealth'-] (3);
    \draw (1) edge[shorten <=2pt, shorten >=2pt, stealth'-] (10);
    \draw (3) edge[shorten <=2pt, shorten >=2pt, stealth'-] (4);
    \draw (3) edge[shorten <=2pt, shorten >=2pt, -stealth'] (5);
    \draw (3) edge[shorten <=2pt, shorten >=2pt, stealth'-] (10);
    \draw (4) edge[shorten <=2pt, shorten >=2pt, -stealth'] (5);
    \draw (4) edge[shorten <=2pt, shorten >=2pt, -stealth'] (7);
    \draw (5) edge[shorten <=2pt, shorten >=2pt, stealth'-stealth'] (7);
    \draw (8) edge[shorten <=2pt, shorten >=2pt, stealth'-] (10);

\end{tikzpicture}%
};
\end{scope}
\draw (current bounding box.south west) rectangle (current bounding box.north east);
\end{tikzpicture}
}{%
\caption{The directed elimination graph $G$ resulting from a specific range
query, with the domination number $\gamma(G)=5$ corresponding to the minimum
number of distance computations needed to separate relevant objects from
irrelevant ones}%
\label{fig:elimgraph}%
}%
\end{floatrow}
\end{figure}
Of course, finding a minimum dominating set is
NP-hard,\footnote{The undirected version is most commonly discussed, with a
reduction, e.g., from set covering~\cite[Th.\,A.1]{Kann:1992}. A similar
reduction to the directed version is straightforward.} and given the rather
unusual clash between large-scale information retrieval and combinatorial
optimization, we may quickly end up with overwhelming instance sizes. Still,
with a suitable mixed-integer programming solver, for example, the
optimization may very well be feasible in many practical cases. As an example,
\cref{fig:experiments} shows some computations made using the Gurobi
solver~\cite{Gurobi:2020}.
Many of these optima were found rather quickly, as presumably the structure of
the elimination graph was amenable to the solution methods of the solver.
Others, such as those for the DNA data set, took several days to compute.
And even for some of the easier cases, there were outliers. For example, for
the 2NN radius in 15-dimensional Euclidean space, all of the 10 randomly
selected queries led to computations lasting 10\kern1pt--200 seconds, except
for one, which took almost twenty hours. As with many such cases, however,
being satisfied with a solution that is a couple of percentage points shy of
perfect could drastically cut down on the computation time (i.e., by setting
the absolute or relative MIP gap), as illustrated in \cref{fig:gapplot}.

\Cref{fig:experiments} also includes results for several other methods, beyond
the optimum. These are all versions of the AESA approach~\cite{Vidal:1986},
as discussed in more depth in \cref{sec:aesa}. At the opposite end of the
spectrum of the optimum, there's the incremental random selection of pivots.
Separating the feasible from the infeasible, is an \emph{oracle} AESA, which
has access to the elimination power of each potential pivot, i.e., how many of
the remaining objects will be eliminated if a given pivot is selected. In the
feasible region we find AESA, iAESA2~\cite{Figueroa:2010}, and the new gAESA,
which is explained in \cref{sec:aesa}.

It is worth noting that $\gamma(G)$ is a more precise lower bound than an
ordinary \emph{best-case} analysis, which only takes input size into account,
and which is therefore always~1. Rather, this is the lowest possible number of
distance computations needed \emph{for a given dataset and query}. In order to
\emph{guarantee} using at most $\gamma(G)$ distance computations, you would
need to somehow determine $G$, which is quite unrealistic. And, as the next
section shows, it is also far from enough.

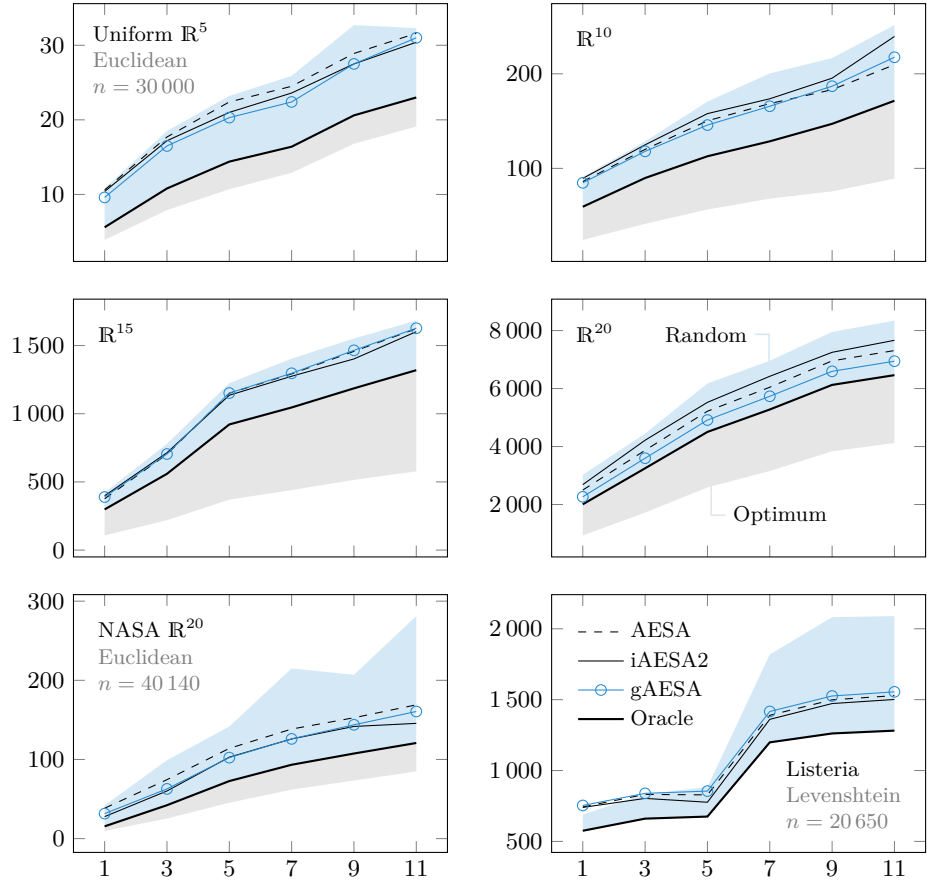
\begin{figure}%
\hfill
\begin{tikzpicture}
\begin{groupplot}[
    group style={
        group name=performance,
        group size=2 by 3,
        vertical sep=0.5cm,
        horizontal sep=1.37cm,
        x descriptions at=edge bottom,
    },
    small,
    xtick={1,3,5,7,9,11},
    height=5cm,
    width=6.5cm
    ]

\mynextgroupplot{uniform_d5_euclidean}
\legend{}

\draw (rel axis cs:0,1) node[inner sep=0pt, below right=7pt] {%
    \scalebox{.9}{\begin{varwidth}{6cm}%
    Uniform $\mathds{R}^{5}$\\%
    \color{black!50}
    Euclidean\\
    $n=\num{30000}$
    \end{varwidth}}%
};

\mynextgroupplot{uniform_d10_euclidean}
\legend{}

\draw (rel axis cs:0,1) node[inner sep=0pt, below right=9pt] {%
    \scalebox{.9}{\begin{varwidth}{6cm}%
    $\mathds{R}^{10}$%
    \end{varwidth}}%
};
\mynextgroupplot{uniform_d15_euclidean}
\legend{}

\draw (rel axis cs:0,1) node[inner sep=0pt, below right=9pt] {%
    \scalebox{.9}{\begin{varwidth}{6cm}%
    $\mathds{R}^{15}$%
    \end{varwidth}}%
};
\mynextgroupplot{uniform_d20_euclidean}
\legend{}

\draw (rel axis cs:0,1) node[inner sep=0pt, below right=9pt] {%
    \scalebox{.9}{\begin{varwidth}{6cm}%
    $\mathds{R}^{20}$%
    \end{varwidth}}%
};

\draw (O) ++(5pt,-4pt)
    node[below right] (opt) {\scalebox{.9}{\footnotesize Optimum}};
\draw[semithick, black!10, cap=rect]
    (O) |- (opt)
    ;

\draw (R) ++(-5pt,4pt)
    node[above left] (rand) {\scalebox{.9}{\footnotesize Random}};
\draw[semithick, sblue!20, cap=rect]
    (R) |- (rand)
    ;

\mynextgroupplot{nasa_euclidean}
\legend{}

\draw (rel axis cs:0,1) node[inner sep=0pt, below right=9pt] {%
    \scalebox{.9}{\begin{varwidth}{6cm}%
    NASA $\mathds{R}^{20}$\\%
    \textcolor{black!50}{Euclidean}\\
    \textcolor{black!50}{$n=\num{40140}$}
    \end{varwidth}}%
};

\mynextgroupplot{listeria_levenshtein}

\draw (rel axis cs:1,0) node[inner sep=0pt, above left=9pt] {%
    \scalebox{.9}{\begin{varwidth}{6cm}%
    Listeria\\%
    \textcolor{black!50}{Levenshtein}\\
    \textcolor{black!50}{$n=\num{20650}$}
    \end{varwidth}}%
};

\end{groupplot}%
\end{tikzpicture}%
\caption{%
Number of distance computations as a function of $k$, the number of nearest
neighbors covered by the chosen radius used for a range query. The first four
datasets are uniformly random vectors, while the last two are taken from the
SISAP dataset collection~\cite{Figueroa:2007}, with queries withheld. The
listeria string lengths vary from~39 to~6579. The results are the average
over~10 randomly selected queries. The oracle AESA uses elimination power
among remaining points as its heuristic}%
\label{fig:experiments}%
\end{figure}%
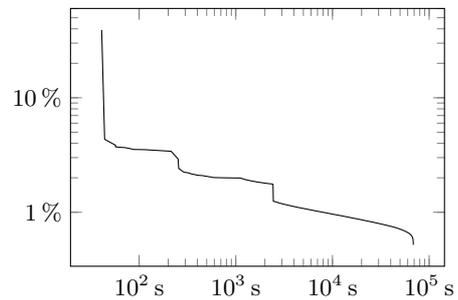
\begin{figure}
\fcapside{%
\caption{%
Bound on relative error (MIP gap) as a function of time, when computing the
optimal number of distance computations in a particularly difficult instance
with $k=2$ over uniformly random vectors in 15-dimensional Euclidean space.
Finding the optimum took over nineteen hours.
After \SI{41}{\second}, the gap was \SI{39.1}{\percent}, but already at
\SI{44}{\second}, it was down to \SI{4.34}{\percent}. Getting to
\SI{1}{\percent} took \SI{2.21}{\hour}
}
\label{fig:gapplot}}{%
\hfill
\begin{tikzpicture}
    \begin{axis}[small, xmode=log, ymode=log,
        height=5cm,
        width=6.5cm,
        ytick = {1,10},
        yticklabels = {\SI{1}{\percent},\SI{10}{\percent}},
        xtick = {100,1000,10000,100000},
        xticklabel style={inner xsep=0pt},
        xticklabels = {
            \SI[parse-numbers=false]{10^2}{\second},
            \SI[parse-numbers=false]{10^3}{\second},
            \SI[parse-numbers=false]{10^4}{\second},
            $\mathclap{\SI[parse-numbers=false]{10^5}{\second}}$,
        },
        ]
    \addplot[mark=none] table [x=secs, y=gap] {gapplot.tsv};
\end{axis}%
\end{tikzpicture}}%
\end{figure}

\section{Metric Search Is Hard, Even If You’re Omniscient}
\label{sec:hardness}

Obviously, a major challenge in choosing the right pivots is that you don't
know what the elimination graph looks like\kern.7pt---you can only make
heuristic guesses. But what if you \emph{did} know? As it turns out, that
wouldn't be the end of your worries.

\Cref{sec:elimination} showed that it is possible to find the optimum by
framing the problem as that of looking for a minimum directed dominating set.
Of course, this is an NP-hard problem, so there's no real surprise in that we
can reduce \emph{to} it. But what about reducing in the other direction? That
is, unless $\mathrm{P}=\mathrm{NP}$, is there any hope of finding some
feasible way of determining the optimum? Alas, no: reducing from the general
minimum undirected dominating set problem to finding the optimum for metric
search is quite straightforward, and the reduction preserves the both the
objective value and the problem size exactly,\footnote{In terms of vertices,
not edges.} meaning that approximation hardness results apply as well.

\begin{theorem}\label{thm:nphard}
    There is a linear-time reduction from the undirected minimum dominating
    set problem on $n$ vertices to the metric range search problem on $n$
    objects, which preserves the objective values of the solutions exactly.
\end{theorem}
\begin{proof}
    \let\d\delta
    We first consider range search.
    To encode any instance $G=(V,E)$ of the minimum dominating set problem, we
    construct a metric space $(X, \d)$, where $X=V\cup\{q\}$, with $q\not\in
    V$, and design the metric so that the elimination graph corresponds to
    $G$. We define the metric as follows:
    $$
    \d(x, y) =
    \begin{cases}
    0 & \text{if $x = y\,$;}\\
    1 & \text{if $\{x, y\}\in E\,$;}\\
    2 & \text{otherwise.}
    \end{cases}
    $$
    In particular, $\delta(q,x)=2$ for all $x\in V$.
    This definition of $\d$ satisfies all the metric properties. Specifically,
    note that triangularity holds, because for any objects $x, y, z\in X$,
    we have $\d(x,z)\leq 2 \leq \d(x,y)+\d(y,z)$ (assuming $x\neq y\neq z$;
    otherwise triangularity is trivial).

    It should be clear that the elimination graph for $q$ with $r<1$
    corresponds exactly to the original graph $G$.\footnote{Note that only the
    lower bound is relevant, as the upper bound is always greater than the
    search radius.} The closed neighborhood $N[x]$ of $x$ (that is, $x$ and
    the set of objects dominated or eliminated by $x$) is $\{y : \d(q,x) -
    \d(x,y) \geq 1\},$ which corresponds exactly to the cases where
    $\d(x,y)=0$ (that is, $x=y$) and where $\d(x,y)=1$ (that is, $\{x,y\}\in
    E$). In other words, any set of pivots that eliminate the remaining
    objects corresponds to a dominating set in $G$, and vice versa. If we find
    such a pivot set of minimum cardinality, we will have solved the
    undirected minimum dominating set problem. In other words, we have a valid
    reduction from the undirected dominating set problem to metric range
    search. It should also be obvious that the reduction can be performed in
    linear time, and that the size of the optimal solutions are
    identical.\footnote{If the new distance is allowed to use the original
    graph as part of its definition, the reduction can be performed in
    \emph{constant} time---it is merely a reinterpretation.}
\end{proof}

\noindent
The previous reduction can be extended to a polytime reduction to $k$NN search
quite easily, showing NP-hardness (though not necessarily preserving
approximation results). We simply set $k=1$ and add another object $\bar x$ so
that $\delta(q,\bar x) = r < 1$ and $\delta(\bar x,y)=2$ for any other object
$y$. Now the minimum $k$NN radius will automatically be~$r$, which gives us
the same reduction as before.

The reduction in the proof of \cref{thm:nphard} constructs a metric range search
problem on $n$ objects from an undirected dominating set problem on $n$ nodes,
so that \emph{if} (and only if) we can solve the search problem (that is, find
a minimum pivot set), we have also solved the minimum dominating set problem.
Approximation bounds thus carry over from the dominating set problem,
so for any $\epsilon>0$, finding solutions that are within a factor of
$(1-\epsilon)\ln n$ is unfeasible, unless
$\mathrm{NP}\subseteq\mathrm{DTIME}(n^{O(\lg\lg n)})$~\cite{Chlebik:2004}.

\begin{corollary}\label{cor:running}
    For instances of the metric range search problem over $n$ objects where the
    optimal number of distance computations is $\gamma$, the worst-case running
    time of any
    algorithm is $\mathrm{\Omega}(\gamma\log n)$, unless
    $\mathrm{NP}\subseteq\mathrm{DTIME}(n^{O(\lg\lg n)})$.
\end{corollary}
\begin{proof}
    An algorithm with a (polynomial) running time of $o(\gamma\log n)$ would
    necessarily use $o(\gamma\log n)$ distance computations, yielding an
    approximation algorithm for the dominating set problem with an
    approximation ratio $o(\log n)$.
\end{proof}

\noindent
Note that the worst-case running time \emph{in general} is still $\Omega(n)$,
as we may very well have $\gamma=n$, in degenerate workloads.

\section{Omniscience Is Overrated}
\label{sec:aesa}

In the discussion so far, what has been described is a scenario where all
potential eliminations are known. Even then, as we have seen, it is only
realistically feasible to get to within a log-factor of the optimum. And as it
turns out, achieving this log-factor is possible, even \emph{without} knowing
all the potential eliminations. What is assumed instead is a more limited
oracle that can tell us which of the remaining points has the highest
\emph{elimination power}, that is, the highest out-degree among the remaining
vertices.

The thing is, a minimum dominating set may be approximated to within a
log-factor using a simple greedy strategy---a strategy that most likely cannot
be significantly improved upon; it gets within a factor of $\ln n +
1$, and as discussed in the previous section, we have a lower bound of
$(1-\epsilon)\ln n$ for any $\epsilon>0$.\footnote{%
The \emph{upper} bound is easily shown by reinterpreting the minimum
dominating set problem for a directed graph $G=(V,E)$ as the problem of
covering $V$ with the closed out-neighborhoods of $G$, translating the
standard set covering approximation~\cite{Williamson:2011}.}

What is more, this is exactly the approach taken by the AESA family of
indexing methods:
they greedily pick one point at a time, based on estimated elimination power,
eliminating others as they go (cf.\@ \cref{fig:cmp}).
\begin{figure}[t]
\begin{tabularx}{\linewidth}{@{}p{.5\linewidth}X@{}}
\begin{pseudo}[dim-color=black!35, font=\footnotesize\kwfont]*
\hd{AESA}(q, r; V, \delta) \\
$U \gets V$ \\
$P = \emptyset;\ \textcolor{\pseudodimcolor}{R=\emptyset}$ \\
while $U\neq\emptyset$ \\+
    $p \gets \arg\min_{x\in U} h_P(x)$ \\[dim]
    if $\delta(q,p)\leq r$\,\tn{:} $R\gets R\cup\{p\}$ \\
    $P \gets P\cup\{p\}$ \\
    $U\gets U\setminus\bigl(\{p\}\cup\{x:\ell_P(x)>r\}\bigr)$ \\-[dim]
return $R$
\end{pseudo}
&
\begin{pseudo}[dim-color=black!35, font=\footnotesize\kwfont]*
\hd{Greedy-Dom-Set}(V, E) \\
$U \gets V$ \\
$D \gets \emptyset$ \\
while $U\neq\emptyset$ \\+
    $p \gets \arg\max_{x\in U} |N^+(x)\cap U|$ \\[dim]
    \\
    $D \gets D \cup \{p\}$ \\
    $U \gets U\setminus N^+[p]$ \\-[dim]
return $D$
\end{pseudo}
\end{tabularx}
\caption{Side-by-side comparison of the AESA metric search algorithm and the
greedy approximation for the directed minimum dominating set
problem}\label{fig:cmp}
\end{figure}
In other words, full omniscience wrt.\@ the elimination graph is not needed;
if we can formulate a heuristic returning the most useful next pivot at each
step, the algorithm is already as good as it realistically can be, or at least
very nearly so.

\begin{proposition}
Greedily selecting pivots based on high elimination power is an asymptotically
optimal polytime strategy for minimizing distance computations in metric range
search, unless $\mathrm{NP}\subseteq\mathrm{DTIME}(n^{O(\lg\lg n)})$.
\qed
\end{proposition}

\noindent
To say that AESA picks pivots based on elimination power may be overstating
it, however. Rather, Vidal Ruiz talks about ``successive
approximation to nearest points''~\cite{Vidal:1986}, while Figueroa et al.\@
state that their goal is ``to define an order such that the first element is
very close to the query,'' because ``[t]he closer the pivot to the query $q$,
the more effective the pruning is''~\cite{Figueroa:2010}.
Of course, all manner of regression and learning methods might be used with
the specific goal of estimating which points are close to the
query~\cite{Edsberg:2010,Mao:2011}, or which are likely to be part of the
search result~\cite{Murakami:2013}.

There has been work on pivot selection focusing directly on elimination
power~\cite{Bustos:2003}, but this does not seem to have been central in
AESA-like methods, using a full distance matrix. One selection method, which
maximizes the lower bound used for elimination, and skips over pivots that
don't contribute, has been explored in the fixed, initial pivot list of
PiAESA~\cite{Socorro:2011}, but the second phase, where pivots are selected
dynamically, still follows the heuristic of selecting those that seem close to
the query.

Following the analogy with the greedy approximation for the directed
dominating set problem, there are two modifications one might make. The first
is to look for high elimination power in the data set overall, rather than
closeness to the query. For example, it is quite possible that a pivot that is
far away might be able to eliminate an entire nearby cluster. The second
modification, which I will briefly explore, is to modify the selection based
on redundancy, i.e., how much of a point's elimination power actually applies
to \emph{remaining} points. If one selects pivots that are as similar to the
query as possible, they are bound to be similar to each other as well; and
even if a pivot is able to eliminate many other points, that is of little use
if those points have already been discarded.

A simple version of this second modification is the following: rather than
merely minimizing the sum of lower bounds, as in AESA, we divide this by the
sum of distances to remaining points. This will not only prefer pivots that
seem to be close to $q$, but those that seem close to $q$ \emph{relative to}
how far they are from the remaining points, meaning they ought to be able to
eliminate more of them. Some preliminary results on the performance of this
\emph{greedy AESA} (gAESA) are shown in \cref{fig:experiments}. As can be
seen, it does seem to perform on par with AESA and iAESA2, at times
outperforming both. Given the rather arbitrary nature of the heuristic, better
variants might very well exist.

\section{Concluding Remarks and Future Work}

The previous sections have established an equivalence between the minimal
number of distance computations needed to resolve an exact metric range
query, on the one hand, and the size of a minimum dominating set in a directed
graph on the other.\footnote{That is, for any range search instance, there is
a directed graph with the objects as its nodes for which the equivalence holds.
Reducing in the other direction preserves the objective value, but not
necessarily the number of nodes/objects.}
The result also applies to uniquely determined $k$NN
queries, if upper bounds are ignored. One might object that the scenario is
too limited---that in practice, one would be contented with an
\emph{approximate} or \emph{probabilistic} search. In fact, the results do
also apply for certain approximations, such as those that merely modify the
query, resulting in a new, simpler exact search~\cite{Naidan:2013}. But
beyond this, the main uses of these results are precisely in establishing the
limits of exact search for given workloads; if one can show that any exact
algorithm must examine an excessively large portion of the data set, that is a
forceful argument in favor of approximation or randomization. What is
presented here only scratches the surface, however. What follows is a sketch
of possible directions for future research based on the established
equivalence.

\textit{Heuristic development.}
The gAESA heuristic is somewhat arbitrary. While it picks pivots that seem
close to the query, relative to the remaining points, the \emph{goal} is to
pick the pivot with the highest elimination power. There may be many ways of
estimating this more directly, either using hand-crafted heuristics (e.g.,
including pivots that are far away from the query compared to remaining
points) or machine learning (which has so far been focused on distance or
relevance).

\textit{Algorithm development.}
In the interest of constructing better baselines, one might take the
development further. Rather than going with the AESA approach, one might
attempt to solve the dominating set problem without actually knowing the
graph. This would be different from the more common forms of online dominating
set problems~\cite{Boyar:2019}, where vertices are provided in some arbitrary
order. Rather, this would presumably involve link prediction~\cite{Lu:2011},
at each step selecting a pivot deemed likely to be included in the optimal
solution or to provide good support for future predictions.

\textit{Problem variants.}
The dominating set problem provides a new perspective on the problem of metric
search, and variants of the former might find analogies for the latter. For
example, the \emph{weighted} dominating set problem can also be approximated
greedily, and the analogous metric search method would be a weighted AESA,
where selection is based on the ratio of weight to elimination power. The
weight could, for example, represent the actual cost of computing the
query--pivot distance, which is the effort that is being minimized, after all.
For many distances, this cost is identical for all points, but for, e.g., the
signature quadratic form distance~\cite{Beecks:2010}, it may vary wildly.

One might also look for analogies in the other direction. For example,
probabilistic methods (such as probabilistic iAESA~\cite{Figueroa:2010}) do
not aim to eliminate all vertices; in these cases, one could instead consider
\emph{partial} domination~\cite{Das:2019}.

\textit{Probabilistic analysis.}
There is a substantial literature on the topic of random graphs.
For example, it is known that for random digraphs whose edges are independent
Bernoulli variables with probability $p$,\footnote{Ch\'avez et al.\@ say that
such independence is a ``reasonable approximation''~\cite{Chavez:2001}.} the
domination number is logarithmic, with base~$1/(1-p)$~\cite{Lee:1998}. In
fact, it is not hard to modify the results of Telelis and
Zissimopoulos~\cite{Telelis:2005} to show that in this scenario, even AESA
selecting pivots \emph{arbitrarily} would yield a logarithmic number of
pivots, staying within a doubly logarithmic additive term of the optimum,
results that match those of Navarro~\cite{Navarro:2009}.

\textit{\kern-1.75ptWorkload descriptions.}
Beyond finding $\gamma$, the dominating set perspective may inspire other
hardness measures and workload descriptions. For example, the greedy
approximation is, more precisely, logarithmic in the \emph{maximum degree}
$\Delta(G)$, a value that could be used as an indicator of \emph{how hard it
is to get close to the optimum}. And although the independence assumption on
elimination may be too strong, one could still use the elimination probability
$p$, perhaps estimated by averaging over several queries, as an indication of
general workload hardness.

\subsubsection*{Acknowledgements.}

The author would like to thank Ole Edsberg, both for discussions
providing the initial idea for this paper, and for substantial later input.
He would also like to thank Jon Marius Venstad and Bilegsaikhan Naidan for
reading early drafts of the paper and providing feedback.

\def\doi#1{\href{http://doi.org/#1}{\nolinkurl{doi:#1}}}
\bibliography{paper}

\begin{thebibliography}{10}
\providecommand{\url}[1]{\texttt{#1}}
\providecommand{\urlprefix}{URL }
\providecommand{\doi}[1]{https://doi.org/#1}

\bibitem{Backurs:2015}
Backurs, A., Indyk, P.: Edit distance cannot be computed in strongly
  subquadratic time (unless {SETH} is false). In: Proceedings of the
  forty-seventh annual ACM symposium on Theory of computing (2015).
  \doi{10.1145/2746539.2746612}

\bibitem{Beecks:2010}
Beecks, C., Uysal, M.S., Seidl, T.: Signature quadratic form distance. In:
  Proceedings of the {ACM} International Conference on Image and Video
  Retrieval. ACM, New York, NY, USA (2010). \doi{10.1145/1816041.1816105}

\bibitem{Boyar:2019}
Boyar, J., Eidenbenz, S.J., Favrholdt, L.M., Kotrb{\v{c}}{\'\i}k, M., Larsen,
  K.S.: Online dominating set. Algorithmica  \textbf{81}(5) (2019).
  \doi{10.1007/s00453-018-0519-1}

\bibitem{Bustos:2003}
Bustos, B., Navarro, G., Ch{\'a}vez, E.: Pivot selection techniques for
  proximity searching in metric spaces. Pattern Recognition Letters
  \textbf{24}(14) (2003). \doi{10.1016/S0167-8655(03)00065-5}

\bibitem{Chavez:2001}
Ch{\'a}vez, E., Navarro, G., Baeza-Yates, R., Marroqu{\'\i}n, J.L.: Searching
  in metric spaces. {ACM} Computing Surveys  \textbf{33}(3) (2001).
  \doi{10.1145/502807.502808}

\bibitem{Chlebik:2004}
Chleb{\'\i}k, M., Chleb{\'\i}kova, J.: Approximation hardness of dominating set
  problems. In: European Symposium on Algorithms. Springer (2004).
  \doi{10.1007/978-3-540-30140-0\_19}

\bibitem{Das:2019}
Das, A.: Partial domination in graphs. Iranian Journal of Science and
  Technology, Transactions A: Science  \textbf{43}(4) (2019).
  \doi{10.1007/s40995-018-0618-5}

\bibitem{Edsberg:2010}
Edsberg, O., Hetland, M.L.: Indexing inexact proximity search with distance
  regression in pivot space. In: Proceedings of the Third International
  Conference on Similarity Search and Applications (2010).
  \doi{10.1145/1862344.1862353}

\bibitem{Figueroa:2007}
Figueroa, K., Navarro, G., Ch\'avez, E.: Metric spaces library (2007),
  available at \url{http://www.sisap.org/Metric_Space_Library.html}

\bibitem{Figueroa:2010}
Figueroa, K., Ch{\'a}vez, E., Navarro, G., Paredes, R.: Speeding up spatial
  approximation search in metric spaces. Journal of Experimental Algorithmics
  \textbf{14} (2010). \doi{10.1145/1498698.1564506}

\bibitem{Ford:1959}
Ford~Jr, L.R., Johnson, S.M.: A tournament problem. The American Mathematical
  Monthly  \textbf{66}(5) (1959). \doi{10.1080/00029890.1959.11989306}

\bibitem{Gurobi:2020}
{Gurobi Optimization, LLC}: Gurobi optimizer reference manual (2020), available
  at \url{http://gurobi.com}

\bibitem{Kann:1992}
Kann, V.: On the Approximability of {NP}-complete Optimization Problems. Ph.D.
  thesis, Department of Numerical Analysis and Computing Science, Royal
  Institute of Technology, Stockholm (1992)

\bibitem{Lee:1998}
Lee, C.: Domination in digraphs. Journal of the Korean Mathematical Society
  \textbf{35}(4) (1998)

\bibitem{Lu:2011}
L{\"u}, L., Zhou, T.: Link prediction in complex networks: A survey. Physica A:
  statistical mechanics and its applications  \textbf{390}(6),  1150--1170
  (2011). \doi{10.1016/j.physa.2010.11.027}

\bibitem{Mao:2011}
Mao, R., Liu, X., Tang, H., Luo, Q., Chen, J., Wu, W.: Multivariate regression
  for pivot selection: A preliminary study. In: 2011 3rd Symposium on Web
  Society. IEEE (2011). \doi{10.1109/SWS.2011.6101281}

\bibitem{Murakami:2013}
Murakami, T., Takahashi, K., Serita, S., Fujii, Y.: Probabilistic enhancement
  of approximate indexing in metric spaces. Information Systems  \textbf{38}(7)
  (2013). \doi{10.1016/j.is.2012.05.012}

\bibitem{Naidan:2013}
Naidan, B., Hetland, M.L.: Shrinking data balls in metric indexes. DBKDA
  (2013)

\bibitem{Navarro:2009}
Navarro, G.: Analyzing metric space indexes: what for? In: Proceedings of the
  2009 Second International Workshop on Similarity Search and Applications.
  SISAP '09, IEEE Computer Society (2009). \doi{10.1109/SISAP.2009.17}

\bibitem{Pestov:2012}
Pestov, V.: Lower bounds on performance of metric tree indexing schemes for
  exact similarity search in high dimensions. Algorithmica  (2013).
  \doi{10.1007/s00453-012-9638-2}

\bibitem{Skopal:2007}
Skopal, T.: Unified framework for exact and approximate search in dissimilarity
  spaces. {ACM} Transactions on Database Systems, {TODS}  \textbf{32}(4)
  (2007). \doi{10.1145/1292609.1292619}

\bibitem{Socorro:2011}
Socorro, R., Micó, L., Oncina, J.: A fast pivot-based indexing algorithm for
  metric spaces. Pattern Recognition Letters  \textbf{32}(11) (2011).
  \doi{10.1016/j.patrec.2011.04.016}

\bibitem{Telelis:2005}
Telelis, O.A., Zissimopoulos, V.: Absolute $o(\log m)$ error in approximating
  random set covering: an average case analysis. Information Processing Letters
   \textbf{94}(4) (2005). \doi{10.1016/j.ipl.2005.02.009}

\bibitem{Traina:2000}
Traina~Jr., C.: Distance exponent: A new concept for selectivity estimation in
  metric trees. In: Proceedings of the 16th International Conference on Data
  Engineering (2000). \doi{10.1109/ICDE.2000.839409}

\bibitem{Vidal:1986}
{Vidal Ruiz}, E.: An algorithm for finding nearest neighbours in
  (approximately) constant average time. Pattern Recognition Letters
  \textbf{4}(3) (1986). \doi{10.1016/0167-8655(86)90013-9}

\bibitem{Williamson:2011}
Williamson, D.P., Shmoys, D.B.: The Design of Approximation Algorithms.
  Cambridge university press (2011)

\end{thebibliography}

\end{document}